\documentclass[letterpaper]{article}
\usepackage[utf8]{inputenc}
\usepackage[T1]{fontenc}
\usepackage{hyperref}
\usepackage{url}
\usepackage{graphicx}
\usepackage{multicol}
\usepackage{subcaption}
\usepackage{tabularx}
\usepackage{amsfonts}
\usepackage{amsthm}
\usepackage{amssymb}
\usepackage{amsmath}
\usepackage{dsfont}
\usepackage{thmtools}
\usepackage{algorithm}
\usepackage{algpseudocode}
\usepackage[colorinlistoftodos,prependcaption,textsize=tiny]{todonotes}
\usepackage{lineno}
\usepackage{sidenotes}
\usepackage{marginnote}
\usepackage{array}
\usepackage{multirow}
\usepackage{authblk}
\usepackage{amsmath}

\newtheorem{theorem}{Theorem}[] 
\newtheorem{lemma}{Lemma}[]

\DeclareMathOperator {\diag}{\text{diag}}

\usepackage[top=2cm, bottom=2cm, left=2cm, right=2cm]{geometry}

\title{Finding Large Independent Sets in Networks Using Competitive Dynamics}
\author[1,2,**]{N.M.~Mooij}
\author[1,2,*]{I.~Kryven}
\affil[1]{Mathematical Institute, Utrecht University, Utrecht, Netherlands}
\affil[2]{Centre for Complex Systems Studies, Utrecht University, Utrecht, Netherlands}
\affil[*]{i.kryven@uu.nl}
\affil[**]{m.n.mooij@uu.nl}
\date{}
\begin{document}

\maketitle

\section{Abstract}
Many decision-making algorithms draw inspiration from the inner workings of individual biological systems. However, it remains unclear whether collective behavior among biological species can also lead to solutions for computational tasks. By studying the coexistence of species that interact through simple rules on a network, we demonstrate that the underlying dynamical system can recover near-optimal solutions to the maximum independent set problem -- a fundamental, computationally hard problem in graph theory. Furthermore, we observe that the optimality of these solutions is improved when the competitive pressure in the system is gradually increased. We explain this phenomenon by showing that the cascade of bifurcation points, which occurs with rising competitive pressure in our dynamical system, naturally gives rise to Katz centrality-based node removal in the network. By formalizing this connection, we propose a biologically inspired discrete algorithm for approximating the maximum independent set problem on a graph. Our results indicate that complex systems may collectively possess the capacity to perform non-trivial computations, with implications spanning biology, economics, and other fields.
\\ \\
{\noindent \small \bf Keywords: Complex Networks, Competing species, Maximum Independent Set, Network Dynamics and Resilience}

\section{Introduction}
Dynamical systems have long been considered as a means of performing analog computations.
Processing information by an artificial neural network is equivalent to a dynamical system becoming stationary at a stable state~\cite{hopfield1982, rutishauser2018solving} or exhibiting a switching behaviour~\cite{AshwinTimme2005}. Revitalised by the recent interest in quantum computing, another possibility is to formulate a computational task as a result of a Hamiltonian optimisation problem. In this case, a solution can be sought as the stationary point of the gradient flow~\cite{brockett1991} or as the outcome of (stochastic) evolution dynamics inspired by the spin glass~\cite{zdeborova2016} and Ising~\cite{fan2023,Lucas2014,mohseni2022,wang2023} models. 
Dynamical systems-based computation is indispensable for approximating problems of computationally hard complexity, with the most promising example of progress so far being clustering and segmentation problems~\cite{vanGennip2019}. Furthermore, ongoing efforts are underway to address other combinatorial optimisation problems within the non-deterministic polynomial (NP) class~\cite{Lucas2014}.

Recent studies in complex networks revealed examples of emergent behaviour in groups of identical agents that interact due to simple rules. The latter can be observed either dynamically, as in synchronisation~\cite{zhang2021}, balancing~\cite{marvel2011},
flocking~\cite{vicsek1995} and melting-like phenomena~\cite{alalwan2019}, or structurally, as in studies of phase transitions in networks~\cite{dSouza2015,sun2023d}. In this paper, we consider an emergent phenomenon in a group of interacting species on a network and show that in the purely competitive regime this system gives a means of approximating a solution to a computationally hard graph-theoretical problem.

The Maximum Independent Set (MIS) problem is a fundamental NP-hard problem~\cite{GareyMR} in computer science, known for numerous applications in, for example, radio network optimization~\cite{ChamaretB}, drug discovery~\cite{meinl2011maximum}, DNA sequencing~\cite{joseph1992determining}, network alignment~\cite{avila2023} or searching a given pattern in a network (subgraph isomorphism). 
A subset of vertices in a graph is called independent when any pair of them is not connected with an edge.
The objective of the MIS problem is to find the largest independent set in a given graph.  Exact algorithms for solving the MIS problem have exponential complexity. 
For this reason, finding ad-hoc fast algorithms is an active area of research~\cite{GuptaM, PelilloMHeuristic, pelillo2006payoff}.

The Lotka-Volterra (LV) system on a network, eminent in ecology, is a system of ordinary differential equations (ODEs) originally introduced to model the evolution and coexistence of interacting biological species~\cite{BuninG,gilpin1973hares, lotka2002contribution}. Currently, other types of interacting agents have also been considered, for example, to show how dynamics of competing firms may lead to an onset of a financial crisis~\cite{comes2012banking, moran2019may}.

One peculiar observation is that the stationary point of the Lotka-Volterra system features a cascade of qualitatively different modes, or bifurcations, when competitive pressure gradually increases. These bifurcations correspond to the events when one of the species becomes extinct. Because of its biological relevance, the first extinction event was referred to as resilience transition~\cite{gao2016} and has been shown to be determined by the structure of the network using spectral graph theory. The crux of our method is to consider the complete cascade of all bifurcation points and show that it ends with the set of surviving species that have a specific graph-theoretical interpretation: 1) they form an \emph{independent} set, and  2) this set has a \emph{maximal} size (\emph{i.e.} it is not necessary the largest one but it cannot be trivially increased without violating the independence property). Furthermore, we show that this principle maximises the output set with effectiveness being comparable to other heuristic algorithms, specifically designed for approximating the largest (maximum) independent set in a graph.

In the rest of the paper we introduce the basic LV system and provide theoretical arguments for why it finds maximal independent sets. Then, we introduce an improved Continuation LV algorithm (CLV) that employs numerical continuation techniques to explore the stable manifold of the LV system and obtain a better approximation of the linear program. We then show that this procedure is equivalent to sequential removal of vertices identified by their Katz centrality. We demonstrate the effectiveness of our algorithms in finding large independent sets by applying them to various networks and comparing with existing benchmark algorithms. 

\section{Results}
Our main result is exemplified through a simple example involving two vertices connected by an edge. Each vertex is associated with the logistic equation for population growth, with growth rate $r>0$, self-inhibition rate $a>0$ and a product coupling strength $\tau>0$:
\begin{align*}
   \begin{cases}
        \frac{\mathrm{d}}{\mathrm{d}t}x_1 &= x_1\left(r - \tau x_2 \right) - a x_1^2,  \\
        \frac{\mathrm{d}}{\mathrm{d}t}x_2 &= x_2\left(r - \tau x_1 \right) - a x_2^2,
    \end{cases}
\end{align*}
where $x_1(t)$ and $ x_2(t)$ represent the concentration of species over time. 
Consider the limiting behaviour at large time, denoted as $x^*_i=\lim\limits_{t\to\infty}x_i(t),$ for $i=1,2$. Local stability analysis reveals that for any $\tau>a$ and distinct initial conditions, $0<x_1(0), x_2(0)<1$, there are exactly two stable stationary states $(0, r/a)$ and $(r/a, 0)$. 
These two states arise when self-inhibition has less influence on population dynamics than pair-wise interactions. 
In low competition scenarios, $\tau<a$, an additional stationary point with positive values emerges.
However, when pair-wise interactions dominate, surviving species cannot coexist with other species in their vicinity. This means that surviving species do not compete with their direct neighbors anymore, leading to stabilisation in the form of uncoupled logistic equations.

When $a$ and $r$ are equal to one, these states resemble the toggle switch: either $x_1^*=1 $ and $ x_2^*=0$, or $x_1^* = 0$ and $x_2^*=1$. 
Using the terminology from logic, the system implements the XOR gate $x_1^* \oplus x_2^* = 1$, as one of the species will asymptotically achieve concentration 1, but never both species simultaneously.
In this example, the states $(1,0)$ and $(0,1)$ correspond to maximal independent sets in the underlying graph.

In general networks, maximal independent sets may  have different sizes, and the related problem of finding the \emph{largest} maximal independent set, maximising the sum $\sum\limits_i x^*_i$ for a given graph, is a known computationally hard problem.
We define the LV system for an arbitrary network with adjacency matrix $A$, 
\begin{equation}\label{eq:sys}
    \frac{\mathrm{d} \mathbf{x}}{\mathrm{d}t} = \textbf{x} \circ ( \textbf{1} - M\textbf{x} ),
\end{equation}
where $\circ$ denotes component-wise multiplication. Here, $\mathbf{x} \in \mathbb{R}^n$ is a vector representing species concentrations, $M = \tau A + \mathbb{I}$ is the interaction matrix, $\mathbb{I}$ denotes the identity matrix, and $\textbf{1}$ is a vector with all components equal to 1.

The number of stationary points of system~\eqref{eq:sys} can vary significantly based on the values of $\tau$ and the underlying network structure. For the trivial graph with no edges, there are $2^n-1$ stable stationary points, while an arbitrary graph and small enough $\tau$ yields only one, globally stable, stationary point $\mathbf{x^*}=M^{-1}\textbf{1}$~\cite{GohBS}. It turns out that in general, as $\tau$ increases, more stable stationary points emerge in this system. When $\tau>1$, the stationary points become binary and satisfy multiple XOR constraints -- one for each edge. 

The connection between independent sets and system~\eqref{eq:sys} becomes apparent when considering the linear integer program formulation for the maximal independent set problem:
\begin{align*}
    &\text{find $x \in \mathbb R^n$} \\
    &\text{such that:} \\
    &x_i (A x)_i = 0 \text{ for all } i \in [n]     &&\text{(A) \textit{Independence}} \\
    &( A + \mathbb{I})\mathbf{x} > 0 \text{ component-wise}  &&\text{(B) \textit{Maximality}} \\
    &x_i( 1 - x_i) = 0                              &&\text{(C) \textit{Binarity}} 
\end{align*}
Here, we deviate from the convention by initially defining $\mathbf{x}$ as a vector of real numbers and later requiring it to be binary with the last constraint. Quite remarkably, satisfying (A) and (C) can be achieved by merely considering the stationary points of dynamical system \eqref{eq:sys} with large enough $\tau$. Furthermore, constraint (B) turns out to be equivalent to local asymptotic stability of these stationary points. Note that the largest $x \in \{0,1\}^n$ that solves this linear integer program is a solution to the maximum independent set problem. 

The link between maximal independent sets and system~\eqref{eq:sys} is made formal by Theorem~\ref{th:lvalgorithm}.
\begin{theorem}\label{th:lvalgorithm}
    Let $A$ be the adjacency matrix of a simple undirected graph. Let $\tau>1$ and let $\mathbf{x}(t)$ be the trajectory of LV system~\eqref{eq:sys}, initialized with initial condition $\mathbf{x_0}$. Then $\mathbf{x^{*}} := \lim\limits_{t\to \infty} \mathbf{x}(t)$ exists, is binary, and the set $\{ v_i \ : \ x^{*}_i = 1 \}$ is a maximal independent set for almost all $\mathbf{x_0} \in (0,1)^n$. 
    For any maximal independent set, its indicator has a basin of attraction of non-zero measure.
\end{theorem}

\noindent In biological terms, this theorem states that under certain  resource limitations, as indicated by $\tau>1$, the LV system reaches the stationary state with a maximal number of species that are not in direct competition, which equivalently corresponds to removing the minimal number of species from the system. 
For a complete proof of Theorem~\ref{th:lvalgorithm} we refer the reader to the Methods section.

Although, for some networks, the system may yield maximal independent sets for smaller $\tau$, the requirement $\tau>1$ is sufficient for any network. The tightness of the bound $\tau>1$ becomes evident when we consider complete graphs. Here, the only non-binary fixed point is represented by a constant vector $\mathbf{x^*}=(\tau(n-1) + 1)^{-1} \mathbf{1}$, and the Jacobian matrix at this point is
\begin{align*}
    J(x^*) &=-
    \begin{pmatrix}
    x^*_i      & \tau x^*_i & \dots  & \tau x^*_i   \\
    \tau x^*_i & x^*_i      &        & \vdots       \\
    \vdots     &            & \ddots & \vdots       \\
    \tau x^*_i & \dots      & \dots  & x^*_i
    \end{pmatrix}.
\end{align*}
By applying Sylvester's criterion to the Jacobian matrix, we conclude that the non-binary stationary point is stable for $\tau < 1$, thereby showing that the bound $\tau > 1$ is tight on the set of all graphs.

Any maximal independent set can be generated from randomly chosen initial conditions whenever $\tau>1$. Additionally, one may improve the approximation by repeating the procedure and searching for the largest set. It is important to acknowledge that this random search methodology may exhibit slow convergence, especially when applied to large graphs. It turns out that the efficacy of this strategy can be enhanced by gradually increasing  parameter $\tau$ in a series of iterative steps. Moreover, each iteration involves re-initializing the dynamical system using the solution obtained from the previous step. This notion is formally realized through the introduction of the Continuation Lotka-Volterra (CLV) algorithm.

\subsection{Continuation Lotka-Volterra algorithm}\label{sec:continuationalgorithm}
We present an algorithm that extends the LV methodology, while eliminating the need for explicit choices on the initial conditions or the parameter $\tau$. The conceptual idea of the algorithm is outlined as follows.

\begin{enumerate}
    \item \emph{Input}: graph $G$  and small parameter $\tau_{\text{step}} >0 $.
    \item Initialise $\tau = 0$ and $\mathbf{x^{*}} = (1, 1, \dots, 1)$.
    \item Increase $\tau$ by $\tau_{\text{step}}$ and run the Lotka-Volterra system \eqref{eq:sys} until convergence to stable equilibrium. \\Set $\mathbf{x^{*}}$ to  be the stationary point.
    \item Keep returning to step 3 until $\tau > 1$.
    \item \emph{Return}: all vertices $v_i$ such that $x_i = 1$.
\end{enumerate}

\noindent The CLV algorithm computes the stable manifold of the LV system while varying the interaction parameter $\tau$ between $0$ and $1$. Initially, at $\tau=0$, all differential equations are decoupled, and their solutions asymptotically converge to $1$ as logistic curves. Upon completion of the algorithm, when $\tau>1$, the solutions converge to an independent set, as indicated by Theorem~\ref{th:lvalgorithm}. Hence, the parameter $\tau$ allows us to anneal between the two regimes, hypothetically increasing the number of non-zero species. This transition is not gradual; instead, it involves the abrupt vanishing of some components of vector $\mathbf{x^*(\tau)}$ due to multiple bifurcation events, which, in turn, results in the final binary structure of the vector $\mathbf{x^*}$. Figure~\ref{fig:numericalcontinuationnumerical} illustrates this behaviour by depicting the values of all steady-state solution variables. In the Methods section, we explain how the algorithm can be implemented efficiently.

\begin{figure}[h!]
    \centering
    \includegraphics[width=\linewidth]{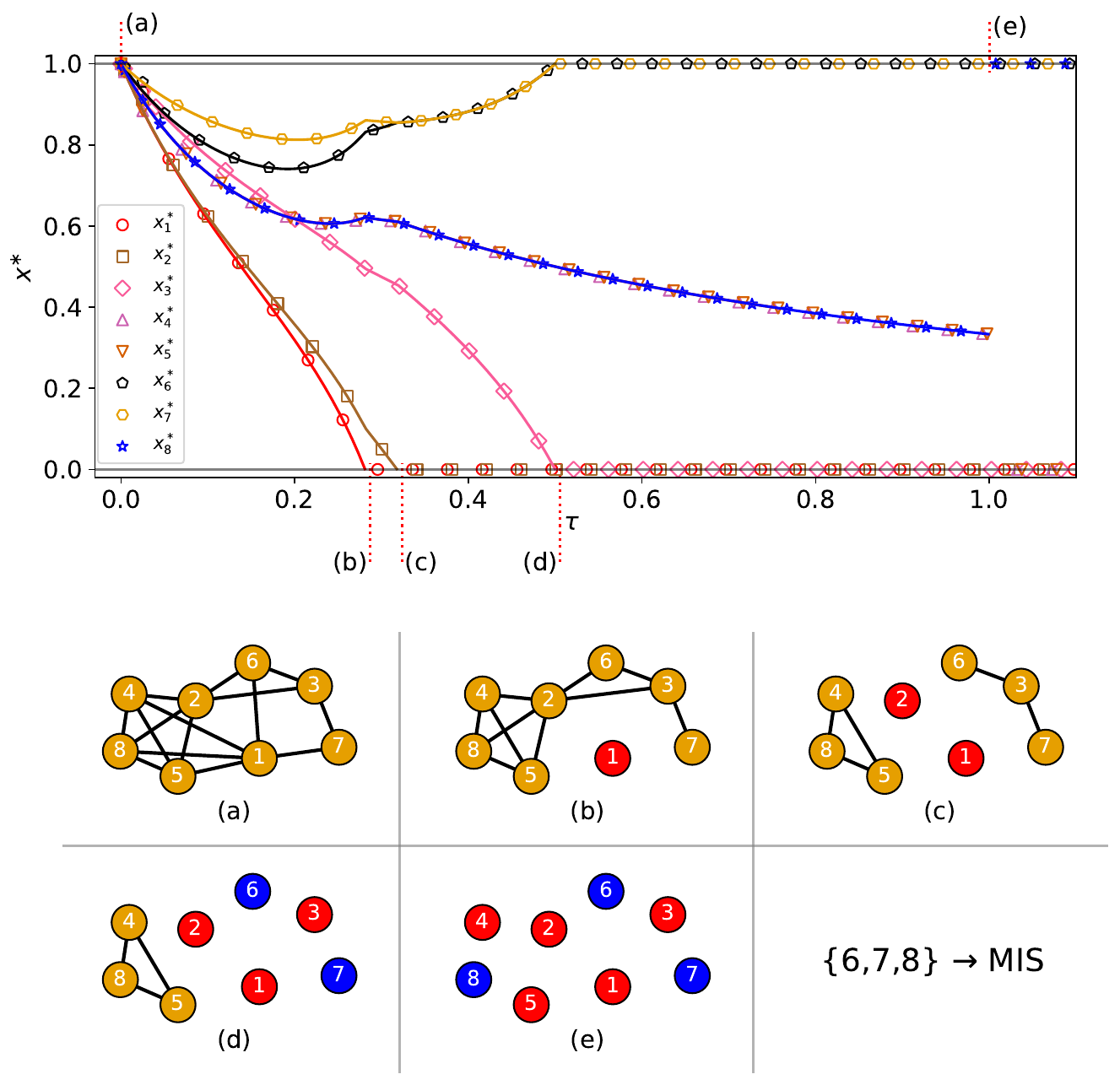}
    \caption{
        Numerical continuation~\cite{allgower_nc} of the stationary point as a function of the parameter $\tau$ in the Lotka-Volterra system reveals an independent set in the interaction network {\bf (a)}.
        As $\tau$ increases starting from 0, the first three bifurcations, indicated by {\bf (b)}, {\bf (c)}, {\bf (d)}, are transcritical bifurcations, while the bifurcation at {\bf (e)} is a pitchfork bifurcation leading to a discontinuous jump in the stationary point. 
        The panels below display the resulting networks after the bifurcations, using the following color coding: \emph{red} vertices vanish after the bifurcation, \emph{orange} vertices are strictly between 0 and 1, and \emph{blue} vertices become 1 after the bifurcation and are thus selected in the maximal independent set.
        }
    \label{fig:numericalcontinuationnumerical}
\end{figure}

Additionally, if the output of either the LV or CLV has size $s$, then at most $d_{\text{max}} = \max\limits_{v \in V} d(v)$ vertices are excluded from the output for each selected vertex.
Consequently, at least $\lfloor (1 + d_{\text{max}})^{-1} |V| \rfloor$ vertices are selected, giving a lower bound on the output of both the LV and CLV algorithm.

\subsubsection{Bifurcation cascade}
The discontinuities and jumps in Figure \ref{fig:numericalcontinuationnumerical} correspond to bifurcation events induced by changes stability of the equilibrium point. There are two possibilities: either the interior equilibrium point moves outside of the domain $[0,1]^n$ and the system reequilibrates at a new point (transcritical bifurcation), or the equilibrium point within the domain loses its local stability (pitchfork bifurcation). Both types of bifurcation result in vanishing if one or more variables, leading to the elimination of a vertex in the underlying graph.  Figure~\ref{fig:numericalcontinuationnumerical}, depicts  a series of transcritical bifurcations followed by a pitchfork bifurcation at the end. Pitchfork bifurcations require special attention when implementing the algorithm numerically -- when this type of bifurcation is detected, one has to randomise the value of $\mathbf{x^*(\tau)}$ by adding a random perturbation of machine-precision magnitude. This step enables the dynamical system to diverge from an unstable stationary solution. We proceed by examining two specific graph types in which we can demonstrate the exact performance of the algorithm.
\\ \\
{\noindent \bf Complete bipartite graph.} It is illustrative to examine the behaviour of the CLV algorithm on a complete bipartite graph, which consists of two distinct vertex sets $X$ and $Y$ with all edges between pairs of vertices in different sets. The maximum independent set is given by either $X$ or $Y$, depending on which vertex set has a greater number of vertices. The computations pertaining to this scenario can be found in the Methods section. The critical value $\tau_{\text{trans}}$ for the occurrence of a transcritical bifurcation is given by
\begin{align*}
    \tau_{\text{trans}} = \min\{ |X|^{-1}, |Y|^{-1} \},
\end{align*}
while for the pitchfork bifurcation we have
\begin{align*}
    \tau_{\text{pitch}} = ( |X| |Y| )^{-1/2}.
\end{align*}

\noindent Since $\tau_{\text{pitch}} > \tau_{\text{trans}}$, unless $|X|=|Y|$, it follows that the CLV algorithm does not encounter pitchfork bifurcations for complete bipartite graphs, unless they have an equal number of vertices in both parts. Consequently, due to the inherent symmetry of the graph, all vertices on the smaller side are simultaneously eliminated, leading to the exact solution of the MIS problem. This contrasts the LV algorithm, which provides an exact answer for all initial conditions only when both parts of the bipartite graph have an equal number of vertices. (This is because maximal independent sets are also maximum sets for well-covered graphs.~\cite{plummer1993well})
\\ \\
{\noindent \bf Path graph.}
We further illustrate the difference between the LV and the CLV algorithm by demonstrating their behaviour on a path graph of length $n$, where every vertex has degree two, except for the outermost vertices that have degree one. The maximum independent sets in such graphs alternate, containing vertices with either all odd or all even indices. In the Methods section, we show that when the total number of vertices $n$ is even, a pitchfork bifurcation occurs prior to any transcritical one. The bifurcation then removes the neighbour of one of the two outermost vertices. That is, the $2$nd or $n-1$th vertex of the path graph, with the pitchfork perturbation deciding which of the two is removed. After such removal, the resulting graph again has an even number of $n-2$ vertices, and the cycle repeats until all vertices are depleted.
When the number of vertices $n$ is odd, the transcritical bifurcation occurs first, at $\tau=1/2$, removing all vertices at even positions, hence $2,4,\dots, n-1$.  Hence, for any $n$, the CLV algorithm recovers an exact MIS. Examples of pre- and post- bifurcation stationary states for odd and even $n$ are illustrated in Figure~\ref{fig:PathGraphFigure}. In contrast to the CLV, the LV algorithm may produce any maximal independent set, including maximal independent sets that have gaps of size two, resulting in suboptimal solutions.

\begin{figure}[h!]
    \centering
    \includegraphics[width=\linewidth]{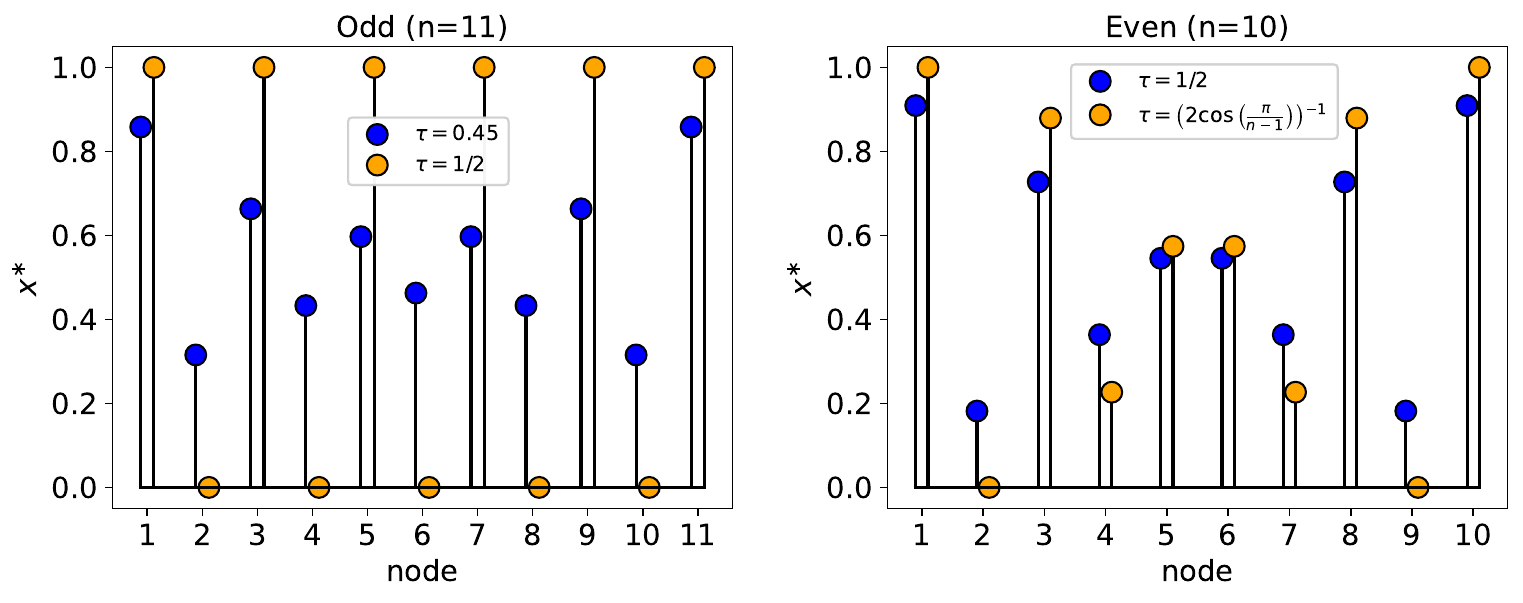}
    \caption{The first bifurcation point in the CLV algorithm on the path graph depends on the parity of $n$, being either pitchfork or transcritical. The values of the steady state $\mathbf{x^{*}}$ plotted for pre- and post-bifurcation values of $\tau$ for odd and even $n$ illustrate this distinction.}
    \label{fig:PathGraphFigure}
\end{figure}

\subsubsection{Link to the Katz centrality}
The Continuation LV algorithm can be interpreted as a greedy algorithm with a sequential elimination strategy, which assesses vertices of a graph $G$ based on their Katz centrality scores, given by
\begin{align}
    C_{\text{Katz}}(v_k) &= \sum_{i=1}^{\infty} \alpha^i \left( (A^i)\mathbf{1} \right)_{k}, \quad \text{for $|\alpha|<\rho(G)$},
\end{align}
where $\rho(G)$ is the spectral radius of $G$.
It turns out that the CLV algorithm resembles ranking nodes with the Katz centrality using negative parameter $\alpha=-\tau<0$. 
Note that the Katz centrality quantifies the influence of vertices by considering their distance from the source node. 
Choosing negative $\alpha$ reproduces this feature and penalises vertices that have a large number of neighbours at odd distances and favours vertices with large number of neighbours at even distance, as illustrated in Figure~\ref{fig:katzcentralitylink}. 
In essence, including the neighbours of neighbours offers an advantage for maximizing the size of an independent set.
A quick way to formalize this intuition is to express the interior equilibrium point in the CLV in terms of walks on a graph (for $\tau$ smaller than the spectral radius of $G$):
\begin{align}
    x^{*}_k &:= (M^{-1} \mathbf{1})_k 
    = \sum_{i=0}^{\infty} \tau^{i} (-1)^{i} \left( A^{i} \mathbf{1} \right)_k. \label{eq:katz}
\end{align}
The term $(-1)^{i}$ ensures that walks of even length are counted as a bonus, while walks of odd lengths are counted as a penalty. 
The first term in the summation in Eq.~\eqref{eq:katz} is equal to one for all $k \in [n]$. This term is absent in the Katz centrality, but because the term is constant for all nodes, it does not affect the vertex ordering.
For every vertex $v_i \in V$, we calculate the corresponding centrality as $\tau_i := \inf\{ \tau \geq 0 \ : \ x^{*}_i = 0 \}$, while also allowing infinity whenever such a $\tau_i$ value does not exist.
\begin{figure}[H]
    \centering
    \includegraphics[width=0.5\linewidth]{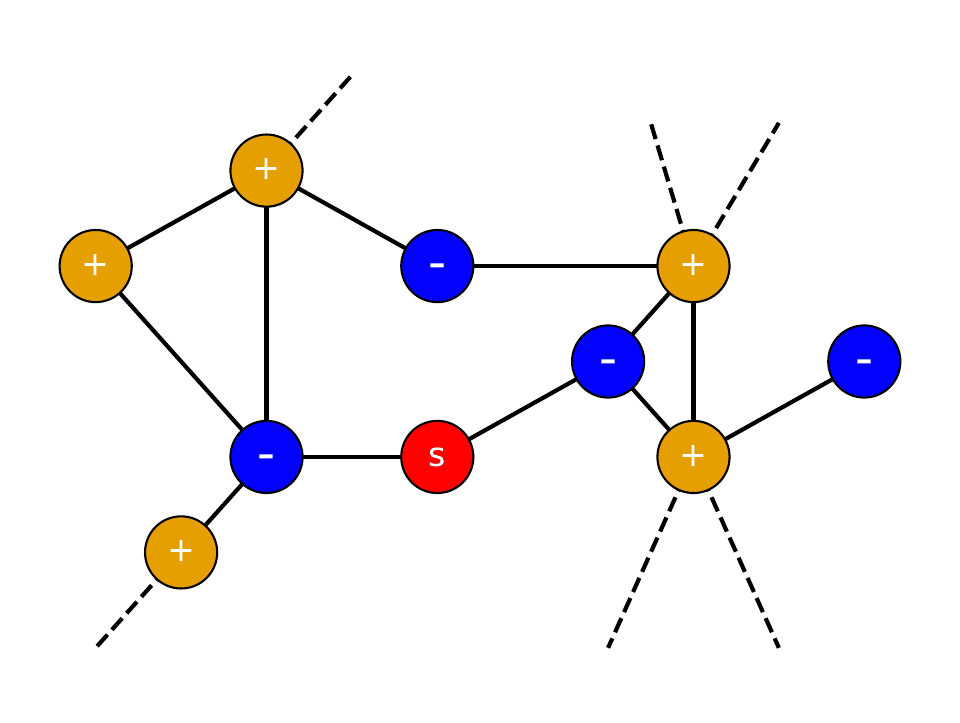}
    \caption{
    For a fixed vertex $s$, its corresponding value of the fixed point receives positive or negative contributions from any other vertex depending on the parity of the shortest paths between this vertex and $s$. The figure illustrates such contributions on an example graph.}
    \label{fig:katzcentralitylink}
\end{figure}
Note that every $j$th term is being scaled with a factor $\tau^j$, prioritizing the significance of first passage over subsequent ones. One can interpret this as close vertices (with regards to shortest path length) having a greater influence on the behaviour of a vertex than vertices that are relatively far away.

\subsection{Numerical results}
Numerical experiments reveal that the independent sets generated by the LV algorithm are strongly biased towards large sizes. 
To investigate this empirically, we utilized the Python PuLP package and the Bron-Kerbosch algorithm~\cite{bron1973algorithm} to compute all maximal independent sets in several random graph models: Erd\H{o}s-Ren\'yi, random bipartite, random geometric graphs, and Barabasi-Albert graphs. 
Figures~\ref{fig:NumericalErdosPerformance}a,~\ref{fig:NumericalBipartitePerformance}a,~\ref{fig:NumericalGeometricPerformance}a, and \ref{fig:NumericalBAPerformance}a demonstrate that the sets generated by the LV algorithm are located at the extreme tail of the size distribution. 
This observation is rather surprising as there is no apparent mechanism for the observed bias. 
Such a pronounced bias is advantageous when performing approximations in practice. 
In fact, the LV system performs similar to popular heuristic algorithms for approximating the MIS problem: Generally, better performance (as measured in terms of approximation factor, percentage of times the output is maximum, and the worst case) than the RPP, Luby and Blelloch algorithms and worse than Minimum Degree Greedy (MDG) algorithm, see Figures~\ref{fig:NumericalErdosPerformance}b,d,e,f,g,h,i, and also the corresponding panels in Figures \ref{fig:NumericalBipartitePerformance}, \ref{fig:NumericalGeometricPerformance}, and \ref{fig:NumericalBAPerformance}. The only exception being the Random Bipartite graph where both Greedy and Blelloch algorithms outperform the LV for some values of the parameter.

Switching to our CLV algorithm, we see even better performance. CLV consistently provides the best results for ER graphs and Barabasi-Albert graphs, and for the Random bipartite graphs the CLV algorithm is comparable with the MDG, with better performance for larger graphs.
This is demonstrated in Figures~\ref{fig:NumericalErdosPerformance}, \ref{fig:NumericalBipartitePerformance}, \ref{fig:NumericalGeometricPerformance}.
On geometric graphs, the CLV outperforms all benchmark algorithms except for MDG, see Figure \ref{fig:NumericalGeometricPerformance}.

We also tested the LV, CLV, MDG, RPP, Luby, and Blelloch's algorithm on graphs from the DIMAC database~\cite{dimacgraphs}. 
The results, detailed in Table~\ref{tab:DIMAC}, demonstrate that the (C)LV performs similar to that of the benchmark algorithms in the majority of cases (out of 51 cases, in 47 cases (C)LV performs the same or better, and in 9 cases the (C)LV algorithm outperforms the rest.).

\begin{figure}[H]
    \centering
    \includegraphics[width=\linewidth]{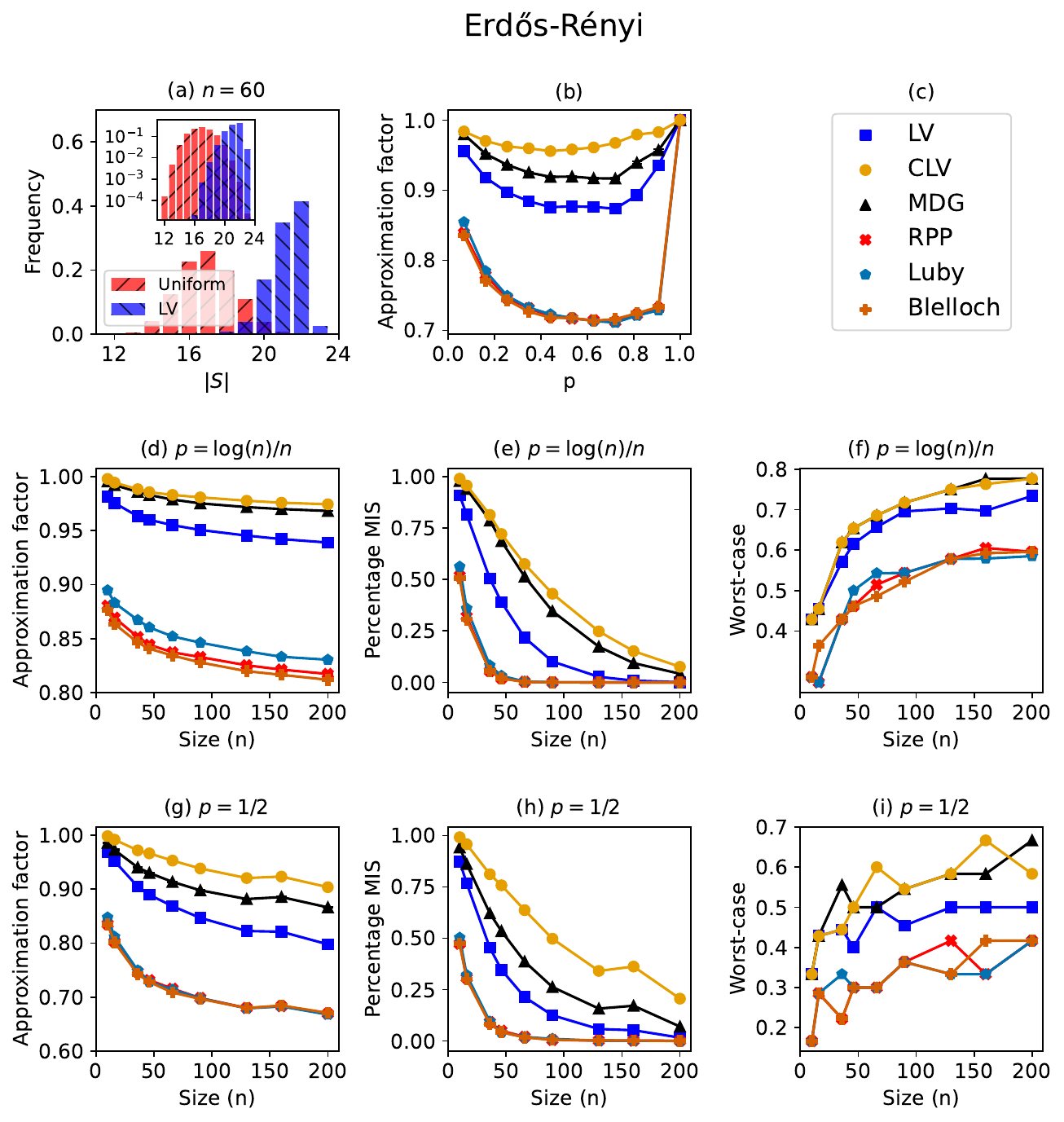}
     \caption{
     The performance of the LV, CLV, MDG, RPP, Luby, and Blelloch algorithm when finding large independent sets in the Erd\H{o}s-Ren\'yi $G(n,p)$ graph; results are averaged over $5000$ runs and the error bars indicate one standard error.
    {\bf(a)} For graphs with $60$ vertices and connection probability $p=0.1$, the total size distribution of maximal independent sets is compared with the distribution generated by the LV algorithm to highlight the algorithm's strong bias towards large sets. The inset features the logarithm of the frequency to emphasise the overlap.
    {\bf(b)} The ratios between the identified and true size of independent sets in graphs with $60$ vertices and different values of parameter $p$. 
    {\bf(d,g)} The ratios between the identified and the true size of independent sets in graphs with a different number of vertices, d) $p=\log(n)/n$ and g) $p=1/2$.
    {\bf(e,h)} The fraction of runs where the algorithm outputs a maximum independent set for graphs with different number of vertices, e) $p=\log(n)/n$ and h) $p=1/2$.
    {\bf(f,i)} The size of the smallest set (worst-case) obtained over all the runs, f) $p=\log(n)/n$ and i) $p=1/2$.
    }
    \label{fig:NumericalErdosPerformance}
\end{figure}

\begin{figure}[H]
    \centering
    \includegraphics[width=\linewidth]{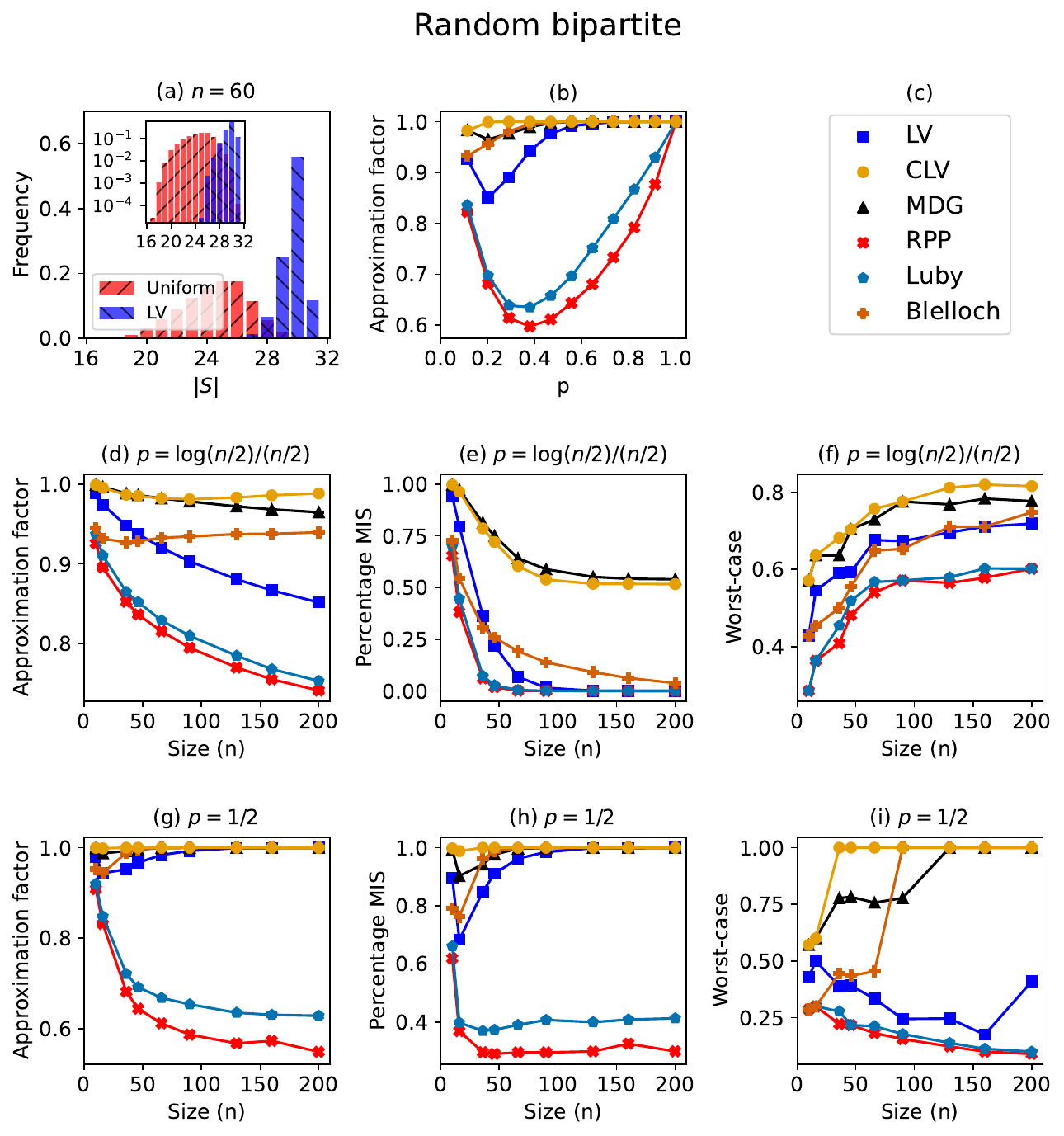}
     \caption{
    The performance of the LV, CLV, MDG, RPP, Luby, and Blelloch algorithm when finding large independent sets in the random bipartite graph; results are averaged over $5000$ runs and the error bars indicate one standard error.
    {\bf(a)} For graphs with $60$ vertices and connection probability $p=0.1$, the total size distribution of maximal independent sets is compared with the distribution generated by the LV algorithm to highlight the algorithm's strong bias towards large sets. The inset features the logarithm of the frequency to emphasise the overlap.
    {\bf(b)} The ratios between the identified and true size of independent sets in graphs with $60$ vertices and different values of parameter $p$. 
    {\bf(d,g)} The ratios between the identified and the true size of independent sets in graphs with a different number of vertices, d) $p=\log(n/2)/(n/2)$ and g) $p=1/2$.
    {\bf(e,h)} The fraction of runs where the algorithm outputs a maximum independent set for graphs with different number of vertices, e) $p=\log(n/2)/(n/2)$ and h) $p=1/2$.
    {\bf(f,i)} The size of the smallest set (worst-case) obtained over all the runs, f) $p=\log(n/2)/(n/2)$ and i) $p=1/2$.
    }    
    \label{fig:NumericalBipartitePerformance}
\end{figure}

\begin{figure}[H]
    \centering
    \includegraphics[width=\linewidth]{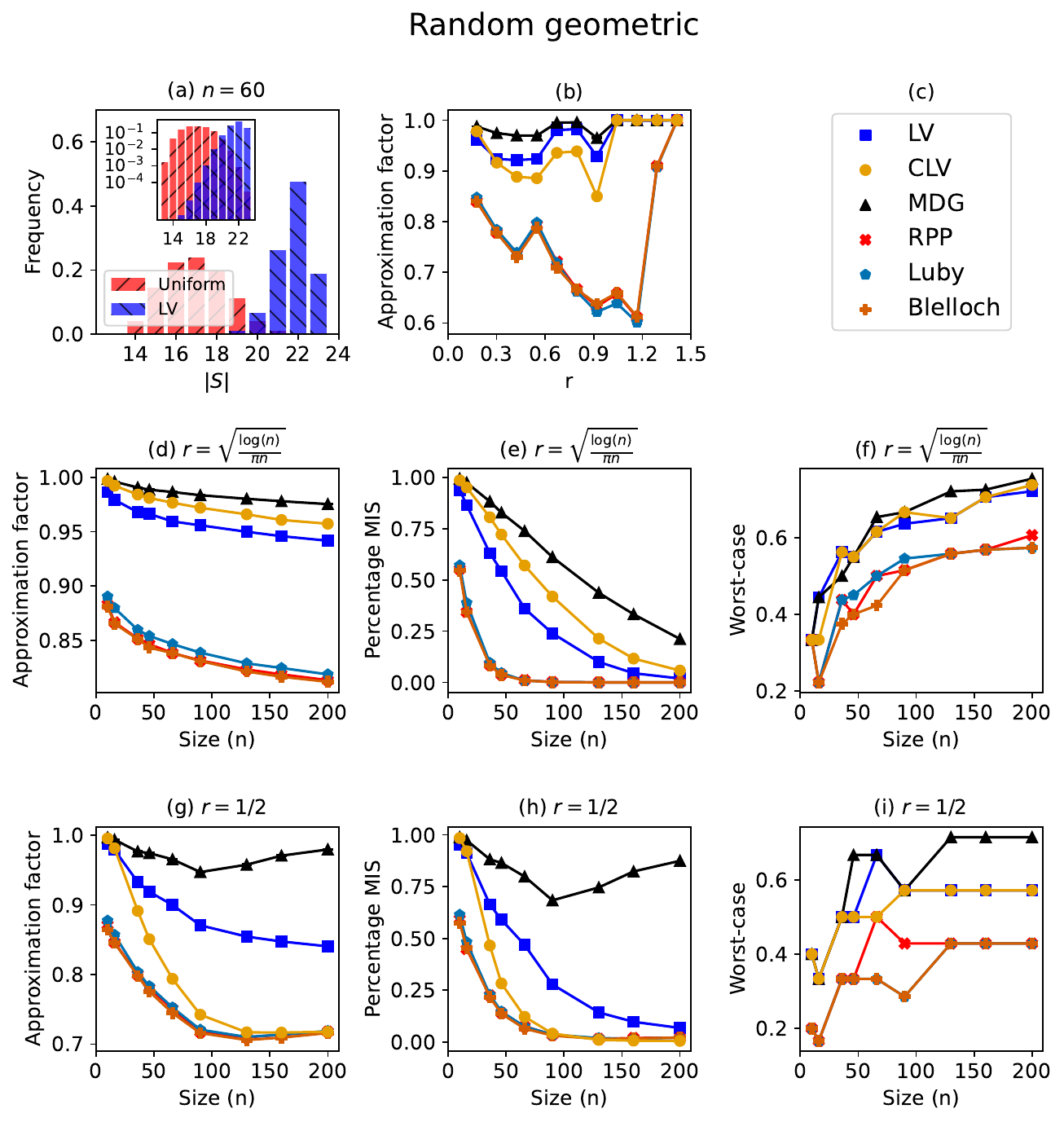}
    \caption{
    The performance of the LV, CLV, MDG, RPP, Luby, and Blelloch algorithm when finding large independent sets in the random geometric graph; results are averaged over $5000$ runs and the error bars indicate one standard error.
    {\bf(a)} For graphs with $60$ vertices and connection radius $r=0.1$, the total size distribution of maximal independent sets is compared with the distribution generated by the LV algorithm to highlight the algorithm's strong bias towards large sets. The inset features the logarithm of the frequency to emphasise the overlap.
    {\bf(b)} The ratios between identified and true size of independent sets in graphs with $60$ vertices and different connection radii $r\in(0,\sqrt{2})$.
    {\bf(d,g)} The ratios between identified and the true size of independent sets in graphs with different number of vertices, d) $r=\sqrt{\frac{\log(n)}{\pi n}}$ and g) $r=1/2$.
    {\bf(e,h)} The fraction of runs that the algorithm outputs a maximum independent set for graphs with different number of vertices, e) $r=\sqrt{\frac{\log(n)}{\pi n}}$ and h) $r=1/2$.
    {\bf(f,i)} The size of the smallest set (worst-case) obtained over all the runs, f) $r=\sqrt{\frac{\log(n)}{\pi n}}$ and i) $r=1/2$.
    }
    \label{fig:NumericalGeometricPerformance}
\end{figure}

\begin{figure}[H]
    \centering
    \includegraphics[width=\linewidth]{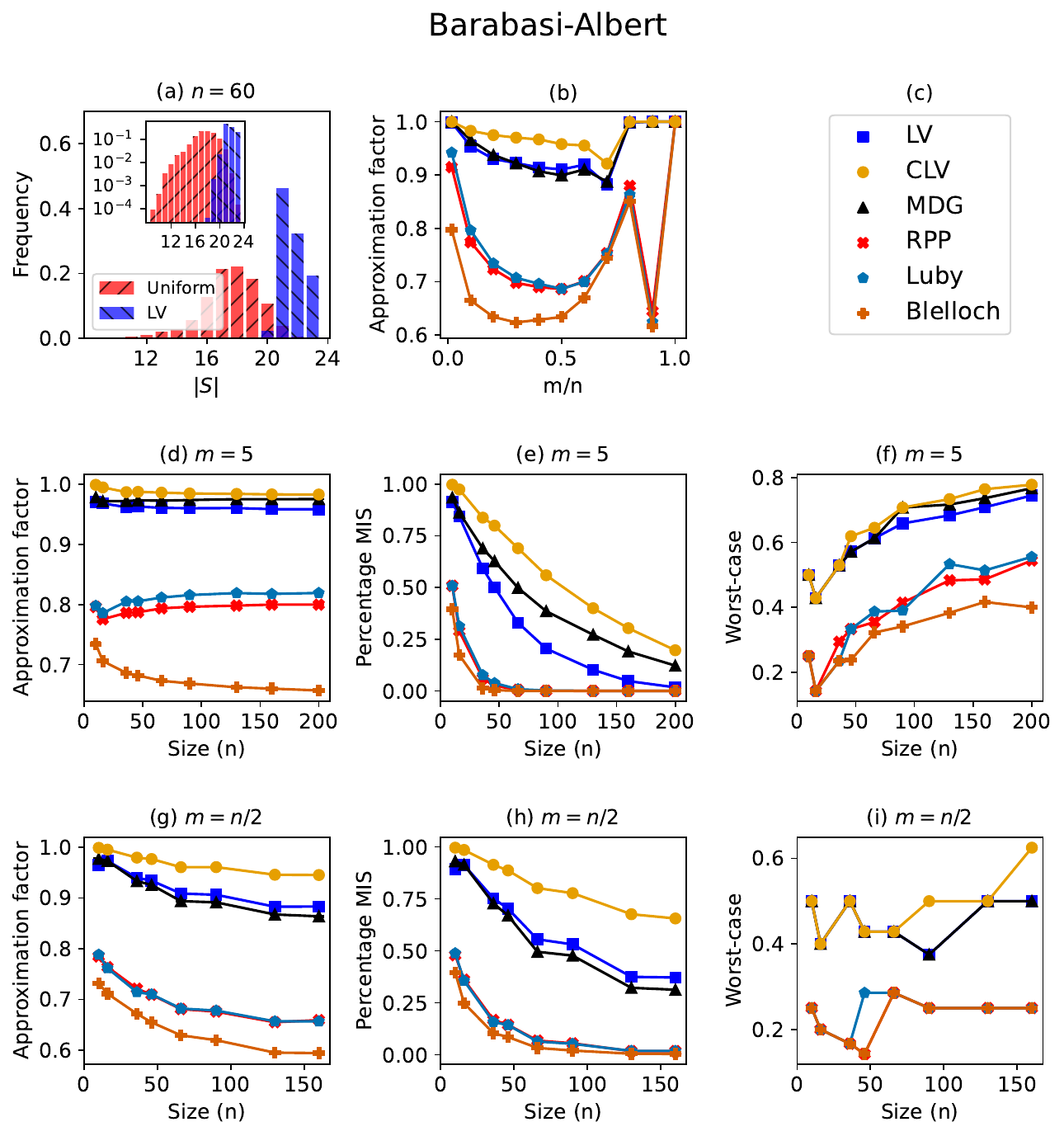}
    \caption{
    The performance of the LV, CLV, MDG, RPP, Luby, and Blelloch algorithm when finding large independent sets in the barabasi-albert graph; results are averaged over $5000$ runs and the error bars indicate one standard error.
    {\bf(a)} For graphs with $60$ vertices and initial number of neighbours $m=5$, the total size distribution of maximal independent sets is compared with the distribution generated by the LV algorithm to highlight the algorithm's strong bias towards large sets. The inset features the logarithm of the frequency to emphasise the overlap.
    {\bf(b)} The ratios between identified and true size of independent sets in graphs with $60$ vertices and fraction of initial number of neighbours $m/n \in(0,1]$.
    {\bf(d,g)} The ratios between identified and the true size of independent sets in graphs with different number of vertices, d) $m=5$ and g) $m=n/2$.
    {\bf(e,h)} The fraction of runs that the algorithm outputs a maximum independent set for graphs with different number of vertices, e) $m=5$ and h) $m=n/2$.
    {\bf(f,i)} The size of the smallest set (worst-case) obtained over all the runs, f) $m=5$ and i) $m=n/2$.
    }
    \label{fig:NumericalBAPerformance}
\end{figure}

\section{Discussion}
We demonstrate a simple mechanism that enables competitive dynamical systems to perform complex computations, namely generating large independent sets on networks. This can be viewed as a new example of computation that can spontaneously emerge in complex systems (natural computing) or be introduced by design (analog computing).
By analysing this mechanism, we proposed two methods for generating large independent sets, which approximate the corresponding NP-hard problem. 

The first approach, the LV algorithm, resembles a gradient descent system, combined with a barrier at $x_i=0$ that ensures the solution remains positive. While there are numerous ways to relax a discrete cost function into a continuous potential, our LV system retains an important favourable property of the original discrete problem: steady-state solutions are binary, and hence the identify subsets of vertices. Moreover, the corresponding continuous optimisation problem has many local minima, which, in turn, identify all maximal independent sets. 

The second approach, the CLV algorithm, can be interpreted as a chain of convex optimisation problems, indexed by values of $\tau \in (0,1]$. Since each of these problems is convex, it admits a unique global solution, which we have also demonstrated to be equivalent to scoring vertices with the Katz centrality. Because of convexity of the individual problems, this is a deterministic algorithm, and moreover, it implements a so-called greedy strategy, removing  vertices one by one, as identified by their centrality. 
This can be interpreted as subjecting the competitive species to gradually increasing pressure of resource depletion. 

Our algorithms show good performance on several popular random graphs. These are either the graphs where the vertices independently connect to each other (the Erd\H{o}sh R\'enyi), and graphs where the degree distribution is a power law originating from the richer-get-richer principle. In both cases, the CLV algorithm showed better performance than the LV algorithm.
Both algorithms preform less well on the random geometric graphs where a simple greedy strategy  outperforms both algorithms.
Sensitivity to underlying geometry of the network, may provide domain restrictions for our algorithms.

Although a realisation of CLV can be interpreted as a greedy algorithm with Katz centrality ranking, the specific type of centrality comes naturally from the LV dynamics as the (transcritical) bifurcation point of this system. Given that the LV system provides a continuous formulation for the maximal independent set problem, our work reveals intricate connection between the Katz centrality and the latter problem.
It is worth mentioning that other dynamical system-inspired centralities include recently proposed DomiRank~\cite{engsig}, which can be seen as a linearised version of the LV system.
Domirank was shown to be useful in partitioning networks, albeit the related bifurcation points are always of the pitchfork type.

In addition to presenting a competitive algorithm for finding large independent sets, the novelty of our approach resides in formulating our algorithm as a dynamical system.  We observe that gradual increase of competitive pressure results in more coexisting species. From a biological perspective, it is evident that when resources are depleted, only those species that do not directly compete will survive. 
However, it is remarkable that under such conditions, the system still tends to maximize the number of surviving species, which is a global property. 
 The paradigm of CLV can also be interpreted as an ecological system that has a tendency to adapt, and hence, dissipate a sudden shock far less effectively than a gradual change. Not having any apparent adaptation mechanism built into the model, our system nevertheless is able to feature such behaviour as an emergent phenomenon. This observation is unexpected and presents a promising direction for future research.

Finally, we highlight several applied areas where our results can be of interest. In computer science, two equivalent problems are finding a minimum vertex cover, the complement of the MIS, and identifying a maximum clique, which corresponds to the MIS in the dual graph. 
Furthermore, the precise detailing of the network will further specialises the problem. For example, using the notion of graph product, our problem becomes equivalent to (sub)graph isomorphism and network alignment problems, prevalent in molecular docking and drug discovery~\cite{gardiner2000}. 
In signal processing, a wide range of questions related to optimal codes can be formulated as the MIS problem~\cite{butenko2002}.
Also, finding large independent sets appears important for wireless sensor networks~\cite{moscibroda2005}. 
Another fundamental challenge in logic, the 3SAT constraint satisfaction problem, can similarly be formulated as the identification of a MIS. 
Interestingly, alternative dynamical system approaches to 3SAT give rise to chaotic dynamics~\cite{ercsey2011optimization}.

\section{Methods}
\subsection{Well-posedness of the LV algorithm}
To show the LV algorithm indeed converges to the claimed result we follow the following steps.
First, we show in Lemma~\ref{lem:escapelemma} that the trajectories are contained in the invariant set $[0,1]^n$ for $t \geq 0$. Second, Lemma~\ref{lem:convergence} states that while being contained in the invariant set, the trajectories converge to isolated points ({\it e.g.} instead of having limit cycles or chaotic behaviour). Third, we show in Lemma~\ref{lem:cornerlemma} that the locally stable stationary solutions must have a binary structure, when $\tau$ is sufficiently large. Finally, Lemma~\ref{lem:maximal} and Lemma~\ref{lem:mis_stable} establish that the latter solutions identify all maximal independent sets. Together, these lemmas prove Theorem~\ref{th:lvalgorithm}. 

\begin{lemma}\label{lem:escapelemma}
    Let $\mathbf{x}(t)$ be the trajectory of system~\eqref{eq:sys} and let $\mathbf{x_0}\in [0,1]^n$ be the corresponding initial condition. Then, $\mathbf{x}(t) \in [0,1]^n$ for all $t \geq 0$.
\end{lemma}
\begin{proof}
    We calculate the flow of the system on the boundary of the unit hypercube $[0,1]^n$ and  show that the trajectories cannot leave this set. Since for any $i=1, \dots, n$ the flow $\frac{\rm{d}}{\rm{d} t}x_i=0$ when $x_i=0$, the trajectories cannot traverse these walls of fixed points and become negative when started inside the set. Furthermore, for arbitrary $i = 1, \dots, n$,
    \begin{align}
        \frac{\mathrm{d} x_i}{\mathrm{d} t} \bigg|_{x_i  = 1} &= x_i[ 1 - (M\mathbf{x})_i ] \bigg|_{x_i = 1}\\
        &= x_i - x_i M_{ii} x_i - x_i \tau \sum_{j\neq i} A_{ij} x_j \bigg|_{x_i = 1}\\
        &= 1 - 1 - \tau \sum_{j\neq i} A_{ij} x_j  \\
        &= - \tau \sum_{j\neq i} A_{ij} x_j \leq 0.
    \end{align}
    Hence, any coordinate of a trajectory cannot exceed one.
\end{proof}

\begin{theorem}[J.~Hofbauer and K.~Sigmund~\cite{hofbauer2003evolutionary}]\label{thm:hofbauer}
    There exists a differentiable, invertible map from $\hat{S}_{n+1} = \{ \mathbf{x} \in S_{n+1} : x_n > 0 \}$ onto $\mathbb{R}_{+}^{n}$ mapping the orbits of the replicator equations with interaction matrix $M'$
    \begin{align}
        \dot{y}_i &= y_i \left( (M'\mathbf{y})_i - \mathbf{y} M' \mathbf{y} \right) \quad i = 1, \dots, n+1
    \end{align}
    onto the orbits of the Lotka-Volterra equations with growth rates $\mathbf{r}$ and interaction matrix $M$
    \begin{align}
        \dot{x}_i &= x_i \left( r_i + \sum_{j=1}^{n-1} M_{ij} x_j \right) \quad i = 1, \dots, n
    \end{align}
    where $r_{i} = M'_{i,n+1} - M'_{n+1,n+1}$ and $M_{ij} = M_{ij}' - M_{n+1,j}'$.
\end{theorem}

\begin{theorem}[E.~Akin and J.~Hofbauer~\cite{akin1982recurrence}]\label{thm:akin}
    Consider the replicator system, $\dot{y_i}= (M'\mathbf{y})_i - \mathbf{y}^T M' \mathbf{y}$ supplied with symmetric matrix $M'$. Let $\mathbf{y}(t)$ be the trajectory of this system, then $\lim\limits_{t\to\infty} \mathbf{y}(t)$ exists.
\end{theorem}

\begin{lemma}\label{lem:convergence}
    System~\eqref{eq:sys} converges to a single equilibrium point that depends on the initial condition.
\end{lemma}
\begin{proof}
    Using the mapping from Theorem~\ref{thm:hofbauer}, the replicator system can be mapped to system~\eqref{eq:sys}.
    Then, Theorem~\ref{thm:akin} claims that all trajectories converge to an equilibrium point.
\end{proof}

\begin{lemma}\label{lem:cornerlemma}
    Let $\mathbf{x}(t)$ be a trajectory of system~\eqref{eq:sys} with initial condition $\mathbf{x}(0)=\mathbf{x_0}$. For some Lebesgue measure zero set $\Omega$ and all $\mathbf{x_0} \in (0,1)^n \setminus \Omega$, the limit $\mathbf{x}^*=\lim\limits_{t\to\infty} \mathbf{x}(t)$ is a binary vector.
\end{lemma}
\begin{proof}
    Let $\mathcal{S}$ represent the set of all stationary points of system~\eqref{eq:sys}.
    Let $\mathbf{x^{*}} \in \mathcal{S}$ be such that $0 < x^{*}_i < 1$ for all components of the vector.
    We first show that the Jacobian evaluated at $\mathbf{x^*}$ has at least one positive eigenvalue, thereby making it a saddle node or an unstable point.
    The Jacobian of the system at this point is given by:
    \begin{align}
        J(\mathbf{x^*}) = \bigg[I - \text{diag}(\mathbf{x}) M - \text{diag}(M\mathbf{x}) \bigg]\bigg|_{\mathbf{x}=\mathbf{x^*}}=-\diag(\mathbf{x^*})M.
    \end{align}
    Since $\mathbf{x^{*}} > 0$, we know that the spectrum of $J(\mathbf{x^{*}})$ is equivalent to the spectrum of the symmetric and real matrix
    \begin{align}
        C(\mathbf{x^{*}}) &= -\diag(\mathbf{x^{*}})^{1/2} M \diag(\mathbf{x^{*}})^{1/2}.
    \end{align}
    Hence, all eigenvalues of $J(\mathbf{x^{*}})$ are real. 
    Moreover, point $\mathbf{x^*}$ is a saddle node if $C(\mathbf{x^*})$ contains at least one positive eigenvalue in its spectrum, which, in turn, is equivalent to $C(\mathbf{x^*})$ having a positive eigenvalue.
    The latter can be tested using the criterion: A symmetric matrix is not negative semi-definite if at least one of its principal minors of even order is negative (or at least one  of its principal minors of odd order is positive), see for example ~\cite{prussing1986principal}.
    \\ \\
    Consider a pair $k, l \in [n]$ with $A_{kl} = 1$. Then, 
    \begin{align}
        \begin{aligned}
            C(x^{*})_{kk} &= -x^{*}_k, \\
            C(x^{*})_{ll} &= -x^{*}_l, \\
            C(x^{*})_{kl} &= -\tau {x^{*}_k}^{1/2} {x^{*}_l}^{1/2} A_{kl}, \\
            C(x^{*})_{lk} &= -\tau {x^{*}_l}^{1/2} {x^{*}_k}^{1/2} A_{lk}.
        \end{aligned}
    \end{align}
    Since $\tau > 1$, the determinant of the $2 \times 2$ leading principal minor of $C$ is given by
    \begin{align}
        C_{kk} C_{ll} - C_{kl} C_{lk} &= x^{*}_k x^{*}_l - \tau^2 x^{*}_k x^{*}_l A_{kl} A_{lk}<0, \label{eq:stabilitysylvesterscriterion}
    \end{align}
    and therefore $C(\mathbf{x^*})$ is not negative semi-definite. Hence, the steady state solution $\mathbf{x^*}$ has at least one positive eigenvalue.
    
    Moreover, note that because the Jacobian at $\mathbf{x^*}$ has eigenvalue zero only when the determinant of $M$ is equal to zero, the Center Manifold Theorem~\cite{perko_manifold} states that the attracting manifold of $\mathbf{x}^*$ always has dimension strictly smaller than $n$.
    We conclude that $\mathbf{x^*}$ attracts trajectories from a manifold $\mathcal{T}_{\mathbf{x^*}}$ of dimension smaller than $n$, \emph{i.e.} a measure-zero set in $[0,1]^n$.
    \\ \\
    By relaxing the lower bound in the assumptions on the stationary point, suppose $0 \leq x^{*}_i < 1$ holds for all $i \in [n]$, with at least one component being strictly positive. We will show that, as before, $\mathbf{x^*}$ cannot be stable. Let $I~=~\{ i \in [n] \ : \ x^{*}_i = 0 \}$ be the set of indices corresponding to zero components. The Jacobian of the Lotka-Volterra system can be rewritten as
    \begin{align}
        J &=
        \begin{pmatrix}
        J' & [-\tau x_i A_{ij}]_{i \notin I, j \in I}\\
        0 & \diag(1 - \tau \sum\limits_{k=1}^n A_{ik} x_k )
        \end{pmatrix}.\label{eq:jacobianequation}
    \end{align}
    The matrix $J'$ denotes the Jacobian of the dynamical system with non-vanishing variables $x_i$, {\it i.e.} with $i \in [n]\setminus I$. Note that $J'$ fulfils the assumptions of the first part of the proof and hence has a positive eigenvalue which is also inherited by $J$.
    \\ \\
    Finally, let us also relax the upper bound. Suppose, $0\leq x^{*}_i \leq 1$ holds for all $i \in [n]$, with at least one component being strictly in $(0,1)$. Let $Q~=~\{ i \in [n] \ : \ x^{*}_i = 1 \}$.
    It follows from the definition of the stationary point that $\frac{{\rm d} x_i}{{\rm d} t}=x^*_i\left(1 - (M\mathbf{x^*})_i\right)=0$, which implies that $(A\mathbf{x^*})_i=0$ for $i\in Q$.
    Using the latter equality for each $i\in Q$, reveals a block diagonal structure of the Jacobian,
    $(J')_{i,i}=-1$ and  $(J')_{i,j}=0$ for $j\neq i$. That is, $|Q|$ blocks of size 1, and one block of size $n-|Q|-|I|$.
    Therefore, without loss of generality, we can exclude all vertices $i \in Q \cup I$ from the system. Let $\mathcal{F}$ denote the set consisting of all the stationary points of system~\eqref{eq:sys} that are not binary. Then the first part of the proof yields the claim with $\Omega = \mathcal{S} \cup \left( \bigcup\limits_{\mathbf{x^*} \in \mathcal{F} } \mathcal{T}_{\mathbf{x^*}} \right)$.
\end{proof}

\begin{lemma}\label{lem:maximal}
    For $\tau>1$ almost all trajectories of system~\eqref{eq:sys} converge to a set indicator that is independent and maximal.
\end{lemma}
\begin{proof}
    From Lemma~\ref{lem:cornerlemma} we know that $\mathbf{x^{*}}$ is binary almost always and hence it identifies a set indicator, $\Gamma := \{ v_i : {{\mathbf x}^{*}_i} = 1 \} \subseteq V$. Suppose $x^*_i = 1$, then $x^*_i \left(1 - (M{{\mathbf x}^*})_i \right) = 0$ implies that $(A \mathbf{x^*})_i = 0$. That is $x^*_j = 0$ for all adjacent $j \sim i$, which shows that $\Gamma$ is indeed an independent set in the graph.
    \\ \\
    Suppose that $\Gamma$ is not maximal. Then, there exists some $i \in [n]$ with $x^*_i = 0$ such that $x^*_j = 0$ for all $j \sim i$. Since $x^*_i = 0$, and $x^*_j = 0$ for all $j \sim i$, they must be contained in the lower right block of matrix $J$, defined in Equation~\eqref{eq:jacobianequation}, and therefore securing an eigenvalue of 
    \begin{align}
        \lambda = 1 - \tau \sum\limits_{k=1}^n A_{ik} x^*_k = 1
    \end{align}
    in the spectrum. Hence, the stable manifold of $\mathbf{x^{*}}$ has dimension smaller than $n$ by the Stable Manifold Theorem, which means it is a measure zero set in $[0,1]^n$. We conclude that $\Gamma$ must be maximal almost always.
\end{proof}

\begin{lemma}\label{lem:mis_stable}
    For $\tau>1$, the indicator of any maximal independent set has a basin of attraction of non-zero measure.
\end{lemma}
\begin{proof}
    Let $\Gamma := \{ v_i : {{\mathbf x}^{*}_i} = 1 \} \subseteq V$ be a maximal independent set. For all $i \in [n]$ with $x^*_i = 0$, there exists $j$ such that $v_i \sim v_j$ and $x^*_j = 1$. Thus, the diagonal elements of the lower right block matrix are negative, and the block structure of the Jacobian in Equation~\eqref{eq:jacobianequation} results in a diagonal matrix with $-1$ on the diagonal in the upper left block. Hence, the Jacobian evaluated for the indicator corresponding to $\Gamma$ has only negative eigenvalues. By the Center Manifold Theorem, this implies that $x^*$ has a stable manifold of non-zero measure.
\end{proof}

\subsection{Analytically solvable graphs}
The maximum eigenvalue of a regular graph is given by the constant degree $d$, corresponding to the eigenvector $\mathbf{1}$. Therefore, for every $\mathbf{x^*}$ satisfying $M\mathbf{x^*} = \mathbf{1}$:
\begin{align}
    \mathbf{x^*} = (\tau d + 1)^{-1} \mathbf{1}.
\end{align}

\noindent Hence, for a complete graph, the interior fixed point is expressed as $x^*_i = (\tau(n-1) + 1)^{-1}$, giving the Jacobian:
\begin{align}
    J(\mathbf{x^*}) &=
    \begin{pmatrix}
    - x^*_i & - \tau x^*_i & \dots & - \tau x^*_i\\
    - \tau x^*_i & - x^*_i & & \vdots\\
    \vdots & & \ddots & \vdots\\
    - \tau x^*_i & \dots & \dots & -x^*_i
    \end{pmatrix}.
\end{align}

\noindent The determinant of $J(\mathbf{x^*})$ is given by:
\begin{align}
    \det(J(\mathbf{x^*})) &= \left( \frac{ \tau-1 }{ \tau(n-1) + 1 } \right)^n \frac{( 1 + \tau (n-1) )}{( 1 - \tau )}.
\end{align}

\noindent Therefore, using the same criterion as in the proof of Lemma~\ref{lem:cornerlemma}, $\mathbf{x^*}$ is always stable for $\tau < 1$. In this case $\tau = 1$ forms a strict bound for the LV algorithm to output a maximal independent set.

\subsection{Implementation of the CLV algorithm}
The CLV algorithm involves the simulation of trajectories for every step, rendering it computationally expensive. To mitigate this issue, we can use a theorem introduced by B.S. Goh~\cite[p.138]{GohBS}. This theorem states that if a nontrivial equilibrium $\mathbf{x^*}$ satisfying $M\mathbf{x^*} = \mathbf{1}$ of the Lotka-Volterra equations is feasible, \textit{i.e.,} $x_i>0$ for all $i \in [n]$, and a constant positive diagonal matrix $C$ exists such that $C M + M^T C$ is negative definite, then the Lotka-Volterra model is globally stable in the feasible region. This implies that in the absence of bifurcation, the interior fixed point $\mathbf{x^*}$ is globally stable. As a result, the CLV algorithm can be modified to exclude domains where no bifurcation has occurred, leading to a substantial reduction in computational cost.

\algrenewcommand\algorithmicrequire{\textbf{Input:}}
\algrenewcommand\algorithmicensure{\textbf{Output:}}

\begin{algorithm}[H]
    \caption*{{\bf CLV algorithm}}
    \begin{algorithmic}[1]
        \Require{Graph $G$ with vertex set $V = (v_i)_{i \in [n] }$, edge set $E$, and adjacency matrix $A$.}
        \Ensure{Maximal independent set of $G$.}
        
        \State $\mathbf{x_{\text{end}}} \leftarrow (1,1, \dots, 1)$.
        
        \While {$|E| > 0$}

            \State $\tau^{*} \leftarrow \inf\{ \tau > 0 \ : \ \lambda_{\text{max}}(J(\mathbf{x^*(\tau)})) = 0 \}$.

            \State $x_{\text{end}} \leftarrow \lim\limits_{t\to\infty}\mathbf{x}(t)$, with $M = \tau^{*} A + \mathbb{I}$ and initial condition $\mathbf{x_{\text{end}}}$.
   
            \State Remove all vertices $v_i$ where $( \mathbf{x_{\text{end}}} )_i = 0$ from $G$.

        \EndWhile

        \State \textbf{return} All vertices $v_i$ such that $(\mathbf{x_{\text{end}}})_i = 0$.
    \end{algorithmic}
    \label{algo:algorithmdisplacement}
\end{algorithm}

\noindent To find the critical values $\tau^{*}$, where these bifurcations occur, we have implemented Newton's method on the functional given by the product of the equilibrium point coordinates, which gives the iterative scheme:
\begin{align}
    \tau_{n+1} = \tau_n - \frac{ \prod\limits_{i \in [n]} x^{*}_{i} }{ \sum\limits_{j \in [n]} (- M^{-1} A M^{-1}\mathbf{1})_{j} \prod\limits_{i \neq j} x^{*}_i }.
\end{align}

\noindent The implemented CLV algorithm can potentially remove too much vertices. This is easily solved by rerunning the system on the graph with all the selected vertices and its neighbours removed.

\subsection{CLV on complete bipartite graphs}\label{app:completebipartitegraph}
Let $G = (V,E)$ be a bipartite graph with $V = X \cup Y$ and $X \cap Y = \emptyset$. We define a complete bipartite graph by setting $d(v) = |Y|$ for every $v \in X$ and $d(v) = |X|$ for every $v \in Y$. We can write the matrix $M$ as 
\begin{align}
    M = 
    -\begin{pmatrix}
        \mathbb{I} & \tau \mathcal{I}_{|X|\times|Y|} \\
        \tau \mathcal{I}_{|Y|\times |X|} & \mathbb{I}
    \end{pmatrix}, \label{eq:completebipartitegraphmatrix}
\end{align}
where $\mathbb{I}$ is the identity matrix, and $\mathcal{I}_{n\times m}$ is a matrix containing only ones of size $n$ times $m$. Using recurrence relationships, the determinant of general $M$ is given by $\det(M) = 1 - |X| |Y| \tau^2$ and therefore the $\tau$ value for which a pitchfork bifurcation occurs is given by $\tau^{*}_{\text{pitch}} = ( |X| |Y| )^{-1/2}$.
\\ \\
For $\tau < ( |X| |Y| )^{-1/2}$, the matrix $M$ is nonsingular, and for $v_k \in X$, we an express $x_k^{*}$ as:
\begin{align}
    x^{*}_k := (M^{-1}\mathbf{1})_k = 1 + \sum_{i=1}^{\infty} (-\tau)^{i} \left( A^{i} \mathbf{1} \right)_k
    = 1 + \sum_{i=1}^{\infty} \tau^{2i} |X|^{i} |Y|^{i} \bigg[ 1 - |X|^{-1} \tau^{-1} \bigg]
    = \frac{ 1 - \tau |Y| }{ 1 - \tau^2 |X| |Y| }.
\end{align}

\noindent The same calculation can be done for $v_k \in Y$, implying $\tau^{*}_{\text{trans}} = \min\{ |X|^{-1}, |Y|^{-1} \}$. Note that $\tau^{*}_{\text{pitch}} \geq \tau^{*}_{\text{trans}}$, and the CLV algorithm is exact.

\subsection{CLV on path graphs}\label{app:odd_sized_linegraph}
The matrix $M = \tau A + \mathbb{I}$ is diagonally dominant for $\tau < 1/2$ since the degree of each vertex is less than or equal to two. Therefore, a pitchfork bifurcation cannot occur before $\tau = 1/2$.
\\ \\
\noindent{\bf Odd size.} For odd-sized path graphs, a transcritical bifurcation takes place at $\tau_{\text{trans}} = 1/2$. Let $x^*_i = 1$ for odd $i$ and $x^*_i = 0$ for even $i$, leading to the equations:
\begin{align}
    \begin{cases}
        x^*_1 + \tau x^*_2 &= 1 \\
        x^*_2 + \tau x^*_1 + \tau x^*_3 &= 1 \\
        \hspace{1cm}\vdots&\hspace{0.1cm}\vdots \\
        x^*_n + \tau x^*_{n-1}&= 1 \\
    \end{cases}
\end{align}
The last equation is satisfied because the size of the graph is odd. Hence, a transcritical bifurcation occurs at $\tau = 1/2$.
\\ \\
\noindent{\bf Even size.} 
We show that a transcritical bifurcation occurs at:
\begin{align}
    \tau^{*} = \left( 2 \cos\left( \frac{\pi}{n-1} \right) \right)^{-1}.
\end{align}

\noindent We use the expression by Hu,~G. and O'Connell,~R.F.~\cite{hu1996analytical}:
\begin{align}
 \begin{split}
 	x^{*}_i &= 
    C^{-1} \sum_{j=1}^{n+1} (-1)^{i+j} \bigg[ \cos( \lambda(n+1 - |j-i|) ) - \cos( \lambda(n+1-i-j) ) \bigg],   \\
    C &= ( -2 \tau \sin(\lambda) \sin((n+1)\lambda) ), \\
	\lambda &= \arccos( (2\tau)^{-1} ).
\end{split}\label{app:eq:taubiggerthan}
\end{align}

\noindent Evaluating $x^{*}_1$ (Equation~\ref{app:eq:taubiggerthan}) at the point $\tau^{*}$ renders $x^{*}_1 = 1$. Moreover, note that:
\begin{align}
    [ x^{*}_i - x^{*}_{i+2} ] \cdot C^{-1} \cdot (-1)^{i} &= \sum_{j}^{n+1} \bigg[ \cos([n+1-|j-i|]) - \cos([ n+1-|j-i-2| ])\bigg] \notag\\
    &- \sum_{j}^{n+1} \bigg[ \cos([n+1-(j-i)]) - \cos([ n+1-(j-i-2) ]) \bigg] \\
     &= (-1)^{i} \bigg[ \cos( [n+1-(n-i)] \lambda ) - \cos( [n+1-(i)] \lambda ) \bigg] \notag\\
     &+ (-1)^{i+1} \bigg[ \cos( [n+1-(n-i+1)] \lambda ) - \cos( [n+1-(i+1)] \lambda ) \bigg] \notag\\
     &+ (-1)^{i} \bigg[ \cos( [n+1-(n+i+2)] \lambda ) - \cos( [n+1-(i+2)] \lambda ) \bigg] \notag\\
     &+ (-1)^{i+1} \bigg[ \cos( [n+1-(n+1+i+2)] \lambda ) - \cos( [n+1-(i+1)] \lambda ) \bigg] \\
    &= (-1)^{i} \cdot (-2) \sin\left(\frac{\lambda}{2}\right) \bigg[ \sin([i+1/2]\lambda) - \sin([i+1/2+1]\lambda) \notag\\
    &+ \sin([n-i-1/2]\lambda) - \sin([n-i-1/2+1]\lambda) \bigg] \\
    &= (-1)^{i} \cdot 2 \sin\left(\frac{\lambda}{2}\right) \cdot 2 \sin\left(\frac{\lambda}{2}\right) \bigg[ \cos\left( \bigg[\frac{2i+2}{2}\bigg] \lambda \right) + \cos\left( \bigg[ \frac{2n-2i}{2}\lambda \bigg] \right) \bigg]\\
    &= (-1)^i \cdot 8 \sin\left(\frac{\lambda}{2}\right)^2 \cos\left(\frac{n+1}{2}\lambda\right) \cos\left(\frac{ n-2i-1 }{ 2 }\lambda\right).
\end{align}

\noindent Substituting $C$ and rewriting we get:
\begin{align}
    x^{*}_{i} - x^{*}_{i+2} &= \frac{1}{\tau} (-1)^{i+1} \tan\left( \frac{\lambda}{2} \right) \cos\left( \frac{ n-2i-1 }{ 2 } \lambda \right) \sin\left( \frac{n+1}{2} \lambda \right)^{-1}.
\end{align}

\noindent On the domain $[0, \tau^{*}]$, the variables $x^*_2$ and $x^*_{n-2}$ are equal and have the lowest values and therefore a transcritical bifurcation occurs at $\tau^{*}~=~\left( 2 \cos\left( \frac{\pi}{n-1} \right) \right)^{-1}$. The determinant of the interaction matrix $M$ with $\tau > 1/2$, obeys the recurrence relation $f(n) = f(n-1) - \tau^2 f(n-2)$, which, together with initial conditions $f(1) = 1$, $f(2) = 1 - \tau^2$, is solved by:
\begin{align}
    f(n) &= 2^{-n} \sum_{l=0}^{\infty} \left( 4 \tau^2 -1 \right)^{l} (-1)^{l} \binom{n+1}{2l+1}.\label{linegraph_recursion}
\end{align}

\noindent Equation~\eqref{linegraph_recursion} at $\tau^{*}$ evaluates to $-2^{1-n} \cos\left( \frac{\pi}{n-1} \right)^{1-n} < 0$ and therefore a pitchfork bifurcation will occur.

\section{Data availability}
All data that support the plots within this paper and other findings of this study are available at \url{https://figshare.com/s/f26a4e23eb2ab892c9dc}.

\section{Code availability}
Code is available for this paper at\\ \url{https://github.com/NiekMooij/Finding-large-independent-sets-in-networks-using-competitive-dynamics}.


\section{Acknowledgements}
The authors are grateful to Yves van Gennip for critical feedback and suggestions. 
NM gratefully acknowledges support from Complex Systems Fund, with special thanks to Peter Koeze. IK gratefully acknowledges support form Netherlands Research Organisation (NWO), research program VIDI, project number VI.Vidi.213.108. 

\section{Author contributions}
IK designed research, NM performed the study, ran simulations and made the figures. NM and IK wrote the manuscript.

\section{Competing interests}
None declared.

\section{Materials \& Correspondence}
Please contact NM for correspondence.

\section{Tables}
\begin{table}[] 
\centering
\begin{tabular}{ccccc|cccccc}
name           & size & $\rho$ & $d_{\text{min}}$ & $d_{\text{max}}$ & LV          & CLV         & MDG      & RPP         & Luby        & Blelloch    \\ \hline
C125-9         & 125  & 0.9    & 102              & 119              & \textbf{4}  & \textbf{4}  & \textbf{4}  & 3           & 3           & 3           \\
C250-9         & 250  & 0.9    & 203              & 236              & 4           & \textbf{5}  & 4           & 3           & 4           & 4           \\
C500-9         & 500  & 0.9    & 431              & 468              & \textbf{4}  & \textbf{4}  & \textbf{4}  & \textbf{4}  & \textbf{4}  & \textbf{4}  \\
MANN-a27       & 378  & 0.99   & 364              & 374              & \textbf{4}  & \textbf{4}  & 2           & 3           & 3           & 2           \\
MANN-a9        & 45   & 0.93   & 40               & 41               & \textbf{3}  & \textbf{3}  & 2           & \textbf{3}  & \textbf{3}  & 2           \\
brock200-1     & 200  & 0.75   & 130              & 165              & 5           & \textbf{6}  & \textbf{6}  & 5           & 5           & 5           \\
brock200-2     & 200  & 0.5    & 78               & 114              & \textbf{10} & 9           & \textbf{10} & 8           & 8           & 8           \\
brock200-3     & 200  & 0.61   & 99               & 134              & \textbf{8}  & \textbf{8}  & \textbf{8}  & 7           & \textbf{8}  & 7           \\
brock200-4     & 200  & 0.66   & 112              & 147              & \textbf{8}  & \textbf{8}  & \textbf{8}  & 6           & 6           & 7           \\
brock400-1     & 400  & 0.75   & 272              & 320              & 6           & \textbf{7}  & 6           & 6           & 5           & 6           \\
brock400-2     & 400  & 0.75   & 274              & 328              & 6           & 6           & \textbf{7}  & 6           & 5           & 6           \\
brock400-3     & 400  & 0.75   & 275              & 322              & \textbf{7}  & 6           & 6           & 6           & 6           & 6           \\
brock400-4     & 400  & 0.75   & 275              & 326              & \textbf{7}  & \textbf{7}  & 6           & 5           & 6           & 5           \\
c-fat200-1     & 200  & 0.08   & 14               & 17               & \textbf{18} & \textbf{18} & \textbf{18} & 17          & 17          & 17          \\
c-fat200-2     & 200  & 0.16   & 32               & 34               & \textbf{9}  & \textbf{9}  & \textbf{9}  & \textbf{9}  & 8           & \textbf{9}  \\
c-fat200-5     & 200  & 0.43   & 83               & 86               & \textbf{3}  & \textbf{3}  & \textbf{3}  & \textbf{3}  & \textbf{3}  & \textbf{3}  \\
c-fat500-1     & 500  & 0.04   & 17               & 20               & 37          & 39          & \textbf{40} & 36          & 36          & 37          \\
c-fat500-10    & 500  & 0.37   & 185              & 188              & \textbf{4}  & \textbf{4}  & \textbf{4}  & \textbf{4}  & \textbf{4}  & \textbf{4}  \\
c-fat500-2     & 500  & 0.07   & 35               & 38               & 19          & 19          & \textbf{20} & 18          & 18          & 19          \\
c-fat500-5     & 500  & 0.19   & 92               & 95               & \textbf{8}  & \textbf{8}  & \textbf{8}  & \textbf{8}  & 7           & \textbf{8}  \\
gen200-p0-9-44 & 200  & 0.9    & 165              & 190              & \textbf{4}  & \textbf{4}  & \textbf{4}  & \textbf{4}  & \textbf{4}  & 3           \\
gen200-p0-9-55 & 200  & 0.9    & 164              & 190              & \textbf{5}  & 4           & 4           & 4           & 3           & 4           \\
gen400-p0-9-55 & 400  & 0.9    & 334              & 375              & 7           & \textbf{8}  & \textbf{8}  & 4           & \textbf{8}  & 4           \\
gen400-p0-9-65 & 400  & 0.9    & 333              & 378              & 5           & 6           & 6           & 5           & \textbf{7}  & 6           \\
gen400-p0-9-75 & 400  & 0.9    & 335              & 380              & \textbf{5}  & \textbf{5}  & \textbf{5}  & \textbf{5}  & \textbf{5}  & 4           \\
hamming6-2     & 64   & 0.9    & 57               & 57               & \textbf{2}  & \textbf{2}  & \textbf{2}  & \textbf{2}  & \textbf{2}  & \textbf{2}  \\
hamming6-4     & 64   & 0.35   & 22               & 22               & \textbf{12} & 8           & \textbf{12} & \textbf{12} & \textbf{12} & \textbf{12} \\
hamming8-2     & 256  & 0.97   & 247              & 247              & \textbf{2}  & \textbf{2}  & \textbf{2}  & \textbf{2}  & \textbf{2}  & \textbf{2}  \\
hamming8-4     & 256  & 0.64   & 163              & 163              & \textbf{16} & \textbf{16} & \textbf{16} & \textbf{16} & \textbf{16} & \textbf{16} \\
johnson16-2-4  & 120  & 0.76   & 91               & 91               & \textbf{15} & \textbf{15} & \textbf{15} & \textbf{15} & \textbf{15} & \textbf{15} \\
johnson32-2-4  & 496  & 0.88   & 435              & 435              & \textbf{31} & \textbf{31} & \textbf{31} & \textbf{31} & \textbf{31} & \textbf{31} \\
johnson8-2-4   & 28   & 0.56   & 15               & 15               & \textbf{7}  & \textbf{7}  & \textbf{7}  & \textbf{7}  & \textbf{7}  & \textbf{7}  \\
johnson8-4-4   & 70   & 0.77   & 53               & 53               & \textbf{5}  & \textbf{5}  & \textbf{5}  & \textbf{5}  & \textbf{5}  & \textbf{5}  \\
keller4        & 171  & 0.65   & 102              & 124              & \textbf{15} & \textbf{15} & \textbf{15} & 14          & 14          & 12          \\
p-hat300-3     & 300  & 0.74   & 168              & 267              & \textbf{8}  & \textbf{8}  & \textbf{8}  & 7           & 7           & 6           \\
p-hat500-2     & 500  & 0.5    & 117              & 389              & \textbf{35} & \textbf{35} & \textbf{35} & 27          & 25          & 22          \\
p-hat500-3     & 500  & 0.75   & 293              & 452              & \textbf{10} & 8           & 9           & 7           & 7           & 6           \\
san200-0-7-1   & 200  & 0.7    & 125              & 155              & \textbf{8}  & \textbf{8}  & 7           & 7           & 7           & \textbf{8}  \\
san200-0-7-2   & 200  & 0.7    & 103              & 164              & \textbf{12} & 11          & \textbf{12} & 10          & 10          & 10          \\
san200-0-9-1   & 200  & 0.9    & 159              & 191              & \textbf{4}  & 3           & \textbf{4}  & 3           & \textbf{4}  & 3           \\
san200-0-9-2   & 200  & 0.9    & 169              & 188              & \textbf{4}  & 3           & \textbf{4}  & \textbf{4}  & \textbf{4}  & \textbf{4}  \\
san200-0-9-3   & 200  & 0.9    & 166              & 187              & \textbf{5}  & \textbf{5}  & \textbf{5}  & \textbf{5}  & \textbf{5}  & 4           \\
san400-0-5-1   & 400  & 0.5    & 174              & 225              & \textbf{32} & \textbf{32} & \textbf{32} & 27          & 25          & 27          \\
san400-0-7-1   & 400  & 0.7    & 257              & 301              & \textbf{11} & 10          & \textbf{11} & 10          & 9           & 10          \\
san400-0-7-2   & 400  & 0.7    & 257              & 304              & 14          & \textbf{15} & 13          & 12          & 11          & 10          \\
san400-0-7-3   & 400  & 0.7    & 250              & 307              & \textbf{19} & \textbf{19} & \textbf{19} & 13          & 15          & 12          \\
san400-0-9-1   & 400  & 0.9    & 341              & 374              & \textbf{5}  & 4           & \textbf{5}  & \textbf{5}  & 3           & 4           \\
sanr200-0-7    & 200  & 0.7    & 120              & 161              & \textbf{7}  & \textbf{7}  & 6           & \textbf{7}  & 6           & 6           \\
sanr200-0-9    & 200  & 0.9    & 166              & 189              & \textbf{4}  & \textbf{4}  & \textbf{4}  & \textbf{4}  & 3           & \textbf{4}  \\
sanr400-0-5    & 400  & 0.5    & 161              & 233              & \textbf{11} & \textbf{11} & \textbf{11} & 9           & 9           & 9           \\
sanr400-0-7    & 400  & 0.7    & 252              & 310              & \textbf{7}  & 6           & \textbf{7}  & \textbf{7}  & 6           & \textbf{7} 
\end{tabular}
\caption{Performance of the LV, CLV, MDG, RPP, Luby, and Blelloch's algorithm on DIMAC graphs. We define the LV, MDG, RPP, Luby, and Blelloch's output to be the largest set obtained in $10$ runs.}
\label{tab:DIMAC}
\end{table}

\end{document}